\newif\iflipics

\iflipics
  \documentclass[a4paper,USenglish,cleveref, autoref]{lipics-v2019}
\else
  \documentclass{article}
  \usepackage{amsfonts,fullpage}
\fi
\usepackage{amssymb,amsmath,amsthm,iftex}
\iflipics\else
  \ifPDFTeX
    \usepackage[noTeX]{mmap}
    \usepackage[T1]{fontenc}
  \fi
  \usepackage[bookmarks]{hyperref}
  \usepackage[nameinlink]{cleveref}
\fi
\title{Lower Bounds for Function Inversion with Quantum Advice}

\newcommand{\kaiminfunding}{Kai-Min Chung is partially supported by the 2019 Academia Sinica Career Development Award under Grant no. 23-17, and MOST QC project under Grant no. MOST 108-2627-E-002-001-.}
\newcommand{\luowenfunding}{Luowen Qian is supported by the DARPA SIEVE program.}
\iflipics
  \author{Kai-Min Chung}{Academia Sinica}{kmchung@iis.sinica.edu.tw}{}{\kaiminfunding}
  \author{Tai-Ning Liao}{National Taiwan University}{b04901168@ntu.edu.tw}{}{}
  \author{Luowen Qian}{Boston University}{luowenq@bu.edu}{}{\luowenfunding}
  
  \authorrunning{K.M. Chung, T.N. Liao, and L. Qian}
  
  \Copyright{Kai-Min Chung, Tai-Ning Liao, and Luowen Qian}
  
  \ccsdesc[500]{Theory of computation~Cryptographic primitives}
  \ccsdesc[500]{Theory of computation~Oracles and decision trees}
  \ccsdesc[500]{Theory of computation~Quantum query complexity}
  \ccsdesc[300]{Theory of computation~Quantum complexity theory}
  
  \keywords{Cryptanalysis, Data Structures, Quantum Query Complexity}
  
  \relatedversion{A full version of the paper is available at \url{https://arxiv.org/abs/1911.09176}.}
  
  \acknowledgements{The authors would like to thank Nai-Hui Chia, Luca Trevisan, Xiaodi Wu, and Penghui Yao for their helpful insights during the discussions. We also thank the anonymous reviewers at QIP and ITC for pointing out various issues in the paper.}
  
  \EventEditors{Yael Tauman Kalai, Adam D. Smith, and Daniel Wichs}
  \EventNoEds{3}
  \EventLongTitle{1st Conference on Information-Theoretic Cryptography (ITC 2020)}
  \EventShortTitle{ITC 2020}
  \EventAcronym{ITC}
  \EventYear{2020}
  \EventDate{June 17--19, 2020}
  \EventLocation{Boston, MA, USA}
  \EventLogo{}
  \SeriesVolume{163}
  \ArticleNo{8}
\else
  \newcommand{\email}[1]{Email: \href{mailto:#1}{\texttt{#1}}}
  \author{
    Kai-Min Chung\footnote{Academia Sinica. \kaiminfunding\ \email{kmchung@iis.sinica.edu.tw}} \and
    Tai-Ning Liao\footnote{National Taiwan University. \email{b04901168@ntu.edu.tw}} \and
    Luowen Qian\footnote{Boston University. \luowenfunding\ \email{luowenq@bu.edu}}
  }
  \date{}
\fi

\iflipics\else
  \newtheorem{theorem}{Theorem}
  \newtheorem{definition}{Definition}
  
  \newtheorem{lemma}{Lemma}
  \newtheorem{claim}{Claim}
  
  \newtheorem{corollary}{Corollary}
\fi
\newtheorem{fact}{Fact}
\def \ws[#1]{\hspace{#1 pt}}
\def \vs[#1]{\vspace{#1 pt}}

\def \Enc {\textnormal{Enc}}
\def \Dec {\textnormal{Dec}}

\def \eps {\varepsilon}

\def \R {\mathcal{R}}

\def \Complex {\mathbb{C}}

\newcommand{\Prob}[1]{\underset{#1}{\textnormal{Pr}}}
\newcommand{\Epc}[1]{\underset{#1}{\mathbb{E}}}
\newcommand{\ket}[1]{|#1\rangle}
\newcommand\Item[1][]{%
  \ifx\relax#1\relax  \item \else \item[#1] \fi
  \abovedisplayskip=0pt\abovedisplayshortskip=0pt~\vspace*{-\baselineskip}}
\begin{document}
\iflipics
  \nolinenumbers
\fi
\maketitle

\begin{abstract}
Function inversion is the problem that given a random function $f: [M] \to [N]$, we want to find pre-image of any image $f^{-1}(y)$ in time $T$.
In this work, we revisit this problem under the preprocessing model where we can compute some auxiliary information or advice of size $S$ that only depends on $f$ but not on $y$.
It is a well-studied problem in the classical settings, however, it is not clear how quantum algorithms can solve this task any better besides invoking Grover's algorithm~\cite{grover1996fast}, which does not leverage the power of preprocessing.

Nayebi et al.~\cite{nayebi2015quantum} proved a lower bound $ST^2 \ge \tilde\Omega(N)$ for quantum algorithms inverting \textit{permutations}, however, they only consider algorithms with \textit{classical} advice.
Hhan et al.~\cite{cryptoeprint:2019:1093} subsequently extended this lower bound to fully quantum algorithms for inverting permutations.
In this work, we give the same asymptotic lower bound to fully quantum algorithms for inverting functions for fully quantum algorithms under the regime where $M = O(N)$.


In order to prove these bounds, we generalize the notion of quantum random access code, originally introduced by Ambainis et al.~\cite{ambainis1999dense}, to the setting where we are given a list of (not necessarily independent) random variables, and we wish to compress them into a variable-length encoding such that we can retrieve a random element just using the encoding with high probability.
As our main technical contribution, we give a nearly tight lower bound (for a wide parameter range) for this generalized notion of quantum random access codes, which may be of independent interest.

\end{abstract}

\section{Introduction}

Space-time trade-offs are a widely observed phenomenon in data structure complexity.
In this work, we are interested in trade-offs between offline preprocessing advice length and online running time in inverting random functions, namely, the trade-off between the size $S$  (in the number of bits) of pre-computed data structure (or advice) on the function (but not the image that we wish to invert) and the algorithm's running time $T$ for computing the inverse of a certain image.
Such trade-offs give lower bound for algorithms that inverts cryptographic functions without taking the specific structure of that family of functions.

Without pre-computed advice ($S = 0$), classical computers require $T = \Omega(\eps N)$ for inverting a random image for a random function $f: [N] \mapsto [N]$ with probability $\eps$, and quantum computers require $T = \Omega(\sqrt{\eps N})$ \cite{ambainis2002quantum} to do so.
Both bounds are asymptotically tight, since we observe that exhaustive search and Grover's algorithm~\cite{grover1996fast} on input range $[\eps N]$ inverts an $\eps$ fraction of inputs, respectively.
However, if we allow some pre-computed advice, classical computers can do much better.
Hellman \cite{hellman1980cryptanalytic} showed that every function can be inverted with $S = T = \tilde O(N^{2/3})$ and every permutation can be inverted using only $S = T = \tilde O(N^{1/2})$.
However, it is not known whether we can do better than Grover's algorithm or Hellman's algorithm, even if we allow quantum computers to come into play.
Therefore motivated by post-quantum cryptanalysis, it is natural to ask whether these two algorithms are indeed the best that we can do.

For classical computers, De et al. \cite{de2009non} (going back to ideas of Yao \cite{yao1990coherent}) showed that $ST = \tilde\Omega(\eps N)$ is required for both functions and permutations, and Corrigan-Gibbs and Kogan \cite{DBLP:journals/eccc/Corrigan-GibbsK18} gave some evidence that improving this lower bound seems to be difficult, by connecting function inversion problem to several other hard problems in complexity theory, communication complexity, etc.
For quantum computers, Nayebi et al. \cite{nayebi2015quantum} showed that $ST^2 = \tilde\Omega(\eps N)$ is required, however, this result only applies to the case where the computation and the oracle queries are quantum but the pre-computed advice remains classical.
However, they also noted that the advice given to a quantum computer can as well be quantum, and it remains open to prove a lower bound for computations in that model.

\subsection{Our Contributions}

In this work, we resolve this discrepancy by showing that $ST^2 = \tilde\Omega(\eps N)$ is still required even if the inverter is allowed to use quantum advice.
Formally,
\begin{definition}
A function (or permutation) inverter is a pair $(\alpha, \mathcal A)$, where:
\begin{enumerate}
    \item $\alpha = \alpha(f)$ is a pre-computed quantum advice of $S$ qubits, which can depend on the function $f: [M] \mapsto [N]$; (for permutations, $M = N$)
    \item $\mathcal A$ is a quantum oracle algorithm that takes advice $\alpha$ and an image $y \in [N]$, makes at most $T$ quantum queries to the function as an oracle $O_f$, and outputs a supposed pre-image $x \in [M]$.
\end{enumerate}
\end{definition}


\begin{definition}
Fix a function inverter $(\alpha, \mathcal A)$. 
\begin{itemize}
    \item We say that ``$(\alpha, \mathcal A)$ inverts $y$ for $f$" if 
    \begin{equation*}
        \Prob{}[f(\mathcal A^{f}(\alpha, y)) = y] \ge 2/3,
    \end{equation*}
    where the probability is taken over the measurement results (internal randomness) of $\mathcal A$.
    \item For any real $\eps$, we say that ``$(\alpha, \mathcal A)$ inverts $\eps$ fractions of inputs" if 
    \begin{equation*}
        \Prob{y, f}[(\alpha, \mathcal A) \textnormal{ inverts $y$ for $f$}] \ge \eps,
    \end{equation*}
    where $y$ and $f$ is sampled uniformly from $[N]$ and $S_N$, respectively.
\end{itemize}

\end{definition}

\begin{theorem}(Lower bound for permutations)
\label{permutations}
For any permutation inverter that invert $\eps$ fractions of inputs, assuming:
\begin{enumerate}
    \Item \begin{equation}
        \label{eq-nondegenerate-eps}
        \eps = \omega(1/N),
    \end{equation}
    that is, the inverter can succeed on more than a constant number of points;
    \Item \begin{equation}
        \label{eq-nondegenerate-T}
        T = o(\eps \sqrt{N}),
    \end{equation}
    noting that $T = O(\sqrt{\eps N})$ is the complexity of Grover's search algorithm;
    \Item \begin{equation}
        \label{eq-nondegenerate-S}
        S \ge 1.
    \end{equation}
\end{enumerate}
We have
\begin{equation*}
    ST^2 \ge \tilde\Omega(\eps N)
\end{equation*}
for all sufficiently large $N$.
\end{theorem}

\begin{theorem}(Lower bound for functions)
\label{functions}
For any function inverter that invert $\eps$ fractions of inputs, assuming:
\begin{enumerate}
    \Item \begin{equation}
        \label{eq-not-too-many-preimages}
        M = O(N),
    \end{equation}
    \Item \begin{equation}
        \label{eq-function-nondegenerate-T}
        T = o(\eps \sqrt{M}/\log^{10} N),
    \end{equation}
    noting that $T = O(\sqrt{\eps M})$ is the complexity of Grover's search algorithm;
    \Item \begin{equation}
        \label{eq-function-nondegenerate-eps}
        \eps \ge 1/N,
    \end{equation}
    that is, the inverter performs no worse than a fixed point output inverter;
    \Item \begin{equation}
        \label{eq-function-nondegenerate-S}
        S \ge 1.
    \end{equation}
\end{enumerate}
We have
\begin{equation*}
    ST^2 \ge \tilde\Omega(\eps M)
\end{equation*}
for all sufficiently large $M$.
\end{theorem}

Towards proving these two theorems, we also develop a lower bound for a natural generalization of quantum random access code (QRAC).
We believe the notion of quantum random access code is a natural object to study in quantum information theory, and that our generalization has potential to find other applications in quantum information.
In \Cref{sec:qrac-vl}, we will explain the concept more thoroughly and prove the lower bound.

\subsection{Related Work}
\label{sec:concurrent-work}


Independently in \cite{cryptoeprint:2019:1093}, they considered a number of cryptographic applications of random functions under both classical advice (quantum query) model and quantum advice model, which they denote as AI-QROM and QAI-QROM respectively.
Under quantum advice model, their Theorem 6 showed bounds for inverting random permutations using different techniques, namely, gentle measurements and semi-classical oracle.

However, in their work, they left the problem of proving bounds for random functions open and we partially give some answers to that open problem in this work.
They noted that generalizing this to function inversion seems problematic -- to use gentle measurement lemma, we need to boost the per-element success probability to $1 - O(1/N^4)$; however, in the function case, even with our idea of using 2-universal hash functions (which we outline in the technical overview section), we cannot hope to boost the per-element success probability beyond $1 - o(1/N)$ as it would already make storing all the hash tags too expensive for an efficient encoding.
In conclusion, it seems hopeless to combine gentle measurement technique with our 2-universal hash for adversaries with constant success probability on $\eps$ fractions of input.
Our QRAC technique, on the other hand, works and gives non-trivial bound even if the per-element success probability is as low as $1 - O(1/\log N)$ under the exact same setting.
This shows that our QRAC technique seems to be able to achieve some improvements compared to their approach.
We also note that our proof technique does not involve internal measurements in the compress/decompress algorithm and is conceptually simpler.


\section{Technical Overview}
\label{sec:tech}

\subsection{Permutations}

We first show how to solve the permutation inversion problem, which is an easier argument.

\paragraph*{Compression argument.}
In De et al. \cite{de2009non}, the main idea in proving the lower bound is to leverage the inverter to produce an algorithm that compresses the permutation into a short string, and the information theoretic lower bound on the size of the string translates to our desired lower bound.
However, as the inverter needs to make $T$ adaptive queries, we need to produce the correct answer for the inverter so that she can successfully invert the image and we can extract the information from the inverter.
The way to do this is to randomly remove a small enough subset of the image from the permutation.
As we are picking a small independently random subset, the probability that the inverter hits this subset will be small.
Therefore, we can use the advice and the permutation without the removed fraction as the encoding for the permutation, and since the length of the encoding is lower bounded by the entropy of all the permutations the encoding scheme is able to compress, this translates to a lower bound in the space-time trade-off for the permutation inversion problem.
In the process, we ``cheated'' by using some shared randomness, but it turns out we can fix this since having shared randomness does not affect the information theoretic lower bound that we need in the end.

As shown by Nayebi et al. \cite{nayebi2015quantum}, this idea also holds similarly for algorithms that can make quantum queries to the permutation.
Namely, if we change $\delta$ fraction of the input, by a similar argument to proving the optimality of Grover's algorithm \cite{ambainis2002quantum}, a quantum query algorithm is required to take $\Omega(\sqrt{1/\delta})$ queries to distinguish the change with constant probability.
However, they also have shown that this approach has a fundamental limitation when one tries to adapt it to the case where the pre-processed information can be quantum.
Recall that in order to invoke the inverter to recover a deleted entry, we need to invoke it with the pre-computed advice.
If the advice is classical, we can simply repeat this process for every entry to recover the entire permutation table; but if the advice is quantum, we cannot hope to do this repeatedly as the previous copy would be destroyed by measurement, and we cannot hope to clone multiple copies of the advice for free due to no cloning theorem \cite{wootters1982single}.
The only thing we can do is to produce multiple copies of the same advice in the encoding phase, however, if we work out the calculation, we can see that this encoding scheme is too inefficient for proving a meaningful lower bound for inverting permutations.

\paragraph*{Avoiding repeated measurements.}
Approaching this challenge, our idea is to reduce the problem to a similar problem that does not require recovering the entire permutation table.
Ambainis et al. \cite{ambainis1999dense} introduced the notion of Quantum Random Access Code with Shared Randomness, which is a two-player game where two players share some randomness $R$; the first player $\mathcal A$ gets a bit string $X$ chosen uniformly at random and is asked to encode it into an encoding $Y \gets \mathcal A(X, R)$; and the second player is asked to recover $X_i$ given $Y, R$ and some index $i \in [|X|]$ chosen uniformly at random.
Assuming the two player succeeds with probability $\delta$, the number of bits in $Y$ is lower bounded by (with some very rough approximations when $\delta \rightarrow 1$) $|Y| \ge \delta |X|$.
It can be shown that this lower bound is tight even when everything is classical, simply by observing that an algorithm that simply remembers a $\delta$ fraction of the input wins the game with probability $\delta$.
This game has found several applications in quantum information theory and quantum cryptography, for example~\cite{alagic2018non}.

Thus, a natural idea is to come up with a similar lower bound for quantum random access code with shared randomness for permutations and do the reduction.
However, unlike in the case of bit strings, as there is correlations between each element of the permutation, our lower bound argument would need to proceed very carefully.
Indeed, in this work we proved a lower bound on the expected number of qubits which is only related to the overall entropy, the average element entropy, and the recovery success probability.
Furthermore, this holds even if there exists correlations between the elements.
In general, this lower bound is weaker than the compression argument where the entire permutation is recovered.
However, we note that if the success probability is high, say $\delta \ge 1 - O(1/N)$ for permutations, then the expected number of qubits needs to be at least $\log N! - O(\log N)$, which asymptotically matches the lower bound for compression argument in the classical case.

A direct encoding scheme would be using the encoding scheme of Nayebi et al. \cite{nayebi2015quantum} and decode only the element in question.
However, this direct idea does not work, since we are randomly removing entries from the permutation, the scheme only succeeds when the removed entries (determined by shared randomness $R$) does not affect the output of the inverter, which only happens with a small probability.
This means that $\delta$ will be bounded away from 1.
Recall that our encoding will need to remember $1 - o(1)$ fraction of the permutation, this gives us no meaningful bound.
In fact, in order for this idea to succeed, we need to boost the success probability to also $1 - o(1)$.

We observe that in our proof for quantum random access code, the length of our encoding is ultimately bounded by the von Neumann entropy of the encoding.
By using the variable length version of quantum source coding theorem, we can also use a variable length encoding that is still bounded by the von Neumann entropy of the encoding.
Specifically, if the randomness will cause the encoding to err, we will simply use the entire permutation table as our encoding, which the decoder can decode any element directly.
By repeating the advice poly-logarithmically many times, we can make the success probability sufficiently close to 1 for proving a meaningful bound.

\subsection{Functions}

To bootstrap the previous argument into an argument for function inverters, we can view the inverse function $f^{-1}$ as a partition of $[M]$, and our goal is to design a random access code for querying this partition.
In order to accommodate all possible adversaries, we only pick the pre-images that have high probability to be returned by the adversary.
However, consider the following bad case, $f^{-1}(y) = \{x_1, x_2\}$, and the adversary uniformly returns $x_1, x_2$ or a third bad output $x'$.
In this case, majority vote will not work since (without loss of generality) assuming we removed $x_1$ from the encoding, the decoder cannot distinguish adversary returns $x_1$ or $x'$ (assuming the adversary gets lucky so that $x'$ is also removed from the encoding).
To fix this, we use a 2-universal hash function (sampled from shared randomness) and use the hash tag to distinguish the correct output.

However, we need to choose the hash length very carefully, as choosing a length too short results in high error probability, and length too long results in inefficient coding (our goal is to achieve nontrivial savings for the random function).
In particular, due to our QRAC bound, we must choose our length tag to be much shorter than $\log N$ to get a nontrivial bound for function inversion.
It turns out that using a length of $\log \log N$ works in our case.

\section{Preliminaries}

We denote $[N]$ to be $\{k\in \mathbb{Z}: 1\le k\le N\}$, and the set of all possible bijections from $[N]$ to itself to be $S_N$.

\begin{definition}(Quantum oracle)
For any classical function $f: X \mapsto Y$ where $Y$ is some additive group, it naturally corresponds to a quantum oracle $O_f$ such that for all $x \in X, y \in Y$,
\begin{equation*}
    O_f(\ket{x} \ket y) = \ket x \ket{y + f(x)}.
\end{equation*}

\end{definition}

Let $\mathcal A^O$ be a quantum oracle algorithm taking $O$ as an oracle.
In the rest of the paper, we will abuse the notation $\mathcal A^f$ to represent $\mathcal A^{O_f}$.
For random oracles, it is equivalent to viewing oracle calls as the same as querying from an exponential sized truth table of the oracle.

\begin{definition}
The query magnitude at $j$ of $\ket\phi = \sum_c \alpha_c \ket c$ is defined to be $q_j(\ket \phi) = \sum_{c \in C_j} |\alpha_c|^2$, where $C_j$ is the set of all computational basis states that query position $j$.
\end{definition}

\begin{definition}
\label{def-total-query-magnitude}
Given a quantum algorithm $\mathcal A$, the total query magnitude at $j$ of $\mathcal A$ with (oracle access to) input $x$ is defined to be $q_j(x) = \sum_{\ket \phi} q_j(\ket \phi)$, where the sum is taken over all the quantum queries produced by the algorithm.
\end{definition}

\begin{lemma}(Swapping lemma) \cite[Lemma 3.1]{vazirani1998power}
\label{swapping-lemma}
Let $\ket{\phi_x}$ and $\ket{\phi_y}$ be the final state of $\mathcal A$ on inputs $x$ and $y$ respectively. Let $T$ be (the upper bound of) the number of queries $\mathcal A$ has made. Then:
\begin{equation*}
    \| \ket{\phi_x} - \ket{\phi_y} \| \le \sqrt{T\sum_{j : x_j \neq y_j}{q_j(x)}},
\end{equation*}
where $\| \ket{\phi_x} - \ket{\phi_y} \|$ denote the Euclidean distance between the two vectors.
\end{lemma}

\begin{theorem} (Quantum Source Coding Theorem) \cite{schumacher2001indeterminate}
	\label{quantum-source-coding-theorem}
	Let $\Sigma$ be an alphabet, $\rho \in D(\Complex^\Sigma)$ be a density operator whose von Neumann entropy is $S(\rho)$.
	\begin{enumerate}
		\item If $L > S(\rho)$, then $N$ independent samples of $\rho$ can be losslessly compressed into $LN$ qubits for all sufficiently large $N$;
		\item If $L < S(\rho)$, then $N$ independent samples of $\rho$ can be losslessly compressed into $LN$ qubits for at most finitely many $N$'s.
	\end{enumerate}
\end{theorem}

\begin{theorem} (2-Universal Hashing)
	\label{thm-hash}
	For every $\eps$, there exists a 2-universal hash function family with error probability $\eps$ and output length $-\log \eps$ (using some finite amount of randomness). \cite[Chapter~3]{vadhan2012pseudorandomness}
\end{theorem}

\section{Quantum Random Access Codes with Variable Length}
\label{sec:qrac-vl}

Intuitively, quantum random access code looks at the following problem:
\begin{itemize}
    \item A random function $f: [N] \mapsto X_N$ is sampled from an \textit{arbitrary} distribution.
    \item At the offline phase, an unbounded algorithm gets access to the entire function and produces a quantum state $\ket\alpha$ of bounded size $\ell$ (therefore dimension at most $2^\ell$).
    \item At the online phase, a uniformly random challenge $x \in [N]$ is generated, and the algorithm given $\ket\alpha$ and $x$ is asked to recover $f(x)$ with probability $\delta$.
\end{itemize}

In this section, we want to prove that there is a trade-off between the expected encoding size $L := \Epc{f}[\ell]$ and the success probability $\delta$.
This is a generalization of QRAC considered in previous works like~\cite{ambainis1999dense} since we can view their QRAC equivalent to ours by making the following restrictions:
\begin{enumerate}
    \item $X_N = \{0, 1\}$.
    \item The function distribution is always the uniform distribution.
    \item The quantum state length $\ell$ is fixed parameter that does not depend on the specific function $f$.
\end{enumerate}

We formalize the problem above as quantum random access code with variable length, as given by the definition below.

\begin{definition}
Let $F_N$ be a set of functions $f: [N] \to X_N$ for some finite set $X_N$.
A quantum random access code with variable length (QRAC-VL) for $F_N$ consists of two algorithms $(\Enc, \Dec)$.
\begin{enumerate}
    \item $\Enc: F_N \times \R \to \Complex^*$. The encoding algorithm encodes a function $f \in F_N$ with some fresh independent randomness in $\R$ to some qubits. The number of qubits denoted by $\ell = \ell(f)$ can depend on the function $f$.
    \item $\Dec: \Complex^* \times [N] \times \R \to X_N$. The decoding algorithm compute $f(x)$ on some specific element $x \in [N]$ with the encoded message in $\Complex^{2^\ell}$, and it uses the same shared randomness for the encoding algorithm.
\end{enumerate}

The performance of the code is measured by two parameters $L$ and $\delta$.
We define
\begin{equation*}
    L := \Epc{f}[\ell(f)]
\end{equation*}
to be the average length of the coding scheme over uniform distribution on $f \in F_N$, and
\begin{equation*}
    \delta := \Prob{f, x, R}[\Dec(\Enc(f; R), x; R) = f(x)]
\end{equation*}
to be the probability that our scheme correctly reconstructs the image of the function, where the probability is taken over uniform distribution on $f \in F_N$, $x \in [N]$, and the scheme's internal randomness.
\end{definition}

First, we prove a helpful lemma that says conditional quantum entropy satisfies subadditivity.

\begin{lemma}
    \label{conditional-entropy-subadditivity-lemma}
	Let $X = (X_1, ..., X_N), Q$ be some quantum states, then
    \begin{equation*}
        \sum_{i = 1}^N S(X_i | Q) \ge S(X | Q).
    \end{equation*}
\end{lemma}
\begin{proof}
We will prove this for $N = 2$ and it is easy to extend this proof to any $N$ using an inductive argument by showing that
$$\sum_{i = 1}^{N - 1} S(X_i | Q) + S(X_N | Q) \ge S(X_1 ... X_{N - 1} | Q) + S(X_N | Q) \ge S(X | Q).$$

For $N = 2$, by the definition of conditional entropy, it is equivalent to prove $S(X_1 Q) + S(X_2 Q) \ge S(X_1 X_2 Q) + S(Q)$, which holds due to strong subadditivity of von Neumann entropy.
\end{proof}

\begin{theorem}(Lower bound for QRAC-VL)
\label{theorem-qrac-vl}
For any QRAC-VL, let $X = (X_1, \dots, X_N)$ be a random variable sampled uniformly random from the distribution (of truth tables) of functions $F_N$.
Therefore, $S(X)$ is the (von Neumann) entropy of a uniformly random distribution of $F_N$ and $S(X_J)$ is the average (or expected) entropy of a single element.
We have that for all sufficiently large $N$,
\begin{equation*}
L \ge S(X) - N \cdot (H(\delta) + (1 - \delta) \cdot S(X_J)),
\end{equation*}
where $H(x) := -x \log_2 x - (1 - x) \log_2 (1 - x)$ is the binary entropy function.
\end{theorem}
\begin{proof}
    Sample $R$ independently.
    Let $Q = \Enc(X; R)$ be the encoding.
    Using the fact in conditional mutual information that $I(Q, R; X) = I(Q; X | R) + I(X; R)$ and the fact that $X$ and $R$ are independent classical random variables,
    \begin{equation}
        \label{eq-qrac-vc-1}
        I(Q, R; X) = I(Q; X | R) \le S(Q | R).
    \end{equation}
    Since $R$ is classical, by \Cref{quantum-source-coding-theorem},
    \begin{equation}
        \label{eq-qrac-vc-2}
        S(Q | R) \le S(Q) \le L.
    \end{equation}

    On the other hand, using \Cref{conditional-entropy-subadditivity-lemma},
    \begin{align}
        \label{eq-qrac-vc-3}
        \begin{split}
		    I(Q, R; X)
    		    &= S(X) - S(X | Q, R) \\
    			&\ge S(X) - \sum_{i=1}^N S(X_i | Q, R) \\
                &= S(X) - N \cdot S(X_J | Q, R, J),
        \end{split}
    \end{align}
    By data processing inequality, we know that
    \begin{align}
    \begin{split}
        \label{eq-qrac-vc-4}
        S(X_J | Q, R, J) \le S(X_J | \Dec(Q, J; R)).
    \end{split}
    \end{align}
    Note that $X_J, \Dec(Q, J; R)$ are both classical random variables.
    Let $I$ be the indicator variable that indicates whether $X_J = \Dec(Q, J; R)$.
    By definition of success probability in quantum random access code, we can show that
    \begin{align}
        \label{eq-qrac-vc-5}
        \begin{split}
            S(X_J | \Dec(Q, J; R))
                &= S(X_J, I | \Dec(Q, J; R)) - S(I | X_J, \Dec(Q, J; R)) \\
                &= S(I | \Dec(Q, J; R)) + S(X_J | I, \Dec(Q, J; R)) - 0 \\
                &\le S(I) + \delta \cdot 0 + (1 - \delta) \cdot S(X_J) \\
                &= H(\delta) + (1 - \delta) S(X_J).
        \end{split}
    \end{align}

    Combining \eqref{eq-qrac-vc-1}, \eqref{eq-qrac-vc-2}, \eqref{eq-qrac-vc-3}, \eqref{eq-qrac-vc-4}, and \eqref{eq-qrac-vc-5}, we get the expected equation in the theorem.
\end{proof}

To see an immediate application of this theorem, we will demonstrate proving a bound for QRAC-VL for permutations.
For permutations, $S(X) = \log N!$ and $S(X_J) = \log N$.
Combining the theorem above with the following algebraic fact, we can prove a lower bound for QRAC-VL for permutations.

\begin{fact}
    $H(1 - \delta) = H(\delta) \le \delta \cdot \log(e/\delta).$
\end{fact}

\begin{corollary}
\label{theorem-qrac-vl-permutation}
For any QRAC-VL for permutations $S_N$ with $\delta = 1 - k/N$ for any $k = \Omega(1/N)$, we have
\begin{equation*}
L \ge \log N! - O(k \log N).
\end{equation*}
\end{corollary}

\section{Proof of \texorpdfstring{\Cref{permutations}}{Theorem for Permutations}}



Now we proceed to construct an encoding scheme given an inverter.
Given a permutation inverter $(\alpha, \mathcal A)$ that inverts an $\eps$ fraction of the input.
Let $\eps' = \eps / 2$.
By how we defined success probability, we can show that there exists a large subset $X$ of all the permutations $S_N$ with size at least $\eps' N!$, such that for any permutation $\pi \in X$, we have that
\begin{equation*}
    \Prob{y}[(\alpha, \mathcal A) \textnormal{ inverts $y$ for $\pi$}] \ge \eps'.
\end{equation*}

Consider a permutation $\pi \in X$, and let $I$ be the set of indices $x \in [N]$ such that $\mathcal A$ inverts $f(x)$. Recall that by the definition of $X$, we have $|I| \ge \eps' N$.
We use the shared randomness in the way such that we sample a subset $R \subseteq [N]$ with each element of $[N]$ independently chosen to be in $R$ with probability $\gamma/T^2$, where $\gamma \in (0, 1)$ is some constant that we will decide later.

Let $G$ be a subset of $I$, where an element $x \in G$ if it satisfies the following two conditions,
\begin{enumerate}
    \Item \begin{equation}
        \label{inverter-condition-1}
        x \in R;
    \end{equation}
    \item The total query magnitude on $R \setminus \{x\}$ while running $\mathcal A^{\pi}(\alpha, \pi(x))$ is bounded by $c/T$ for some constant $c$, that is,
    \begin{equation}
        \label{inverter-condition-2}
        \sum_{z\in R\setminus\{x\}}{q_z(x)}\le \frac c T.
    \end{equation}
\end{enumerate}

\begin{claim}
    With probability at least 0.8 over the choice of $R$, $|G| = \Omega(\eps N/T^2)$.
\end{claim}
\begin{proof}
    Let $H = R \cap I$.
    Due to the definition of $R$, $|H|$ is distributed according to a binomial distribution.
    Therefore, the expected value of $|H|$ is $|I| \gamma/T^2$.
    By the multiplicative Chernoff bound and \eqref{eq-nondegenerate-T},
    \begin{equation}
        \label{eq-size-H}
        \Prob{R}\left[|H| \ge \frac{|I|\gamma}{2T^2}\right] \ge 0.9
    \end{equation}
    for all sufficiently large $N$.

    By definition, each query that $\mathcal A$ makes is of unit length.
    Since $\mathcal A$ makes at most $T$ queries, by \Cref{def-total-query-magnitude},
    \begin{equation*}
        \sum_{z \in [N]} q_z(x) \le T.
    \end{equation*}
    By linearity of expectation,
    \begin{equation*}
        \Epc{R}\left[\sum_{z \in R \setminus \{x\}}q_z(x) \right]
            = \sum_{z \in [N] \setminus \{x\}} \frac \gamma {T^2} q_z(x)
            \le \frac \gamma{T^2} T
            = \frac \gamma T.
    \end{equation*}
    Hence, by Markov's inequality,
    \begin{equation}
        \label{eq-prob-J}
        \Prob{R}\left[\sum_{z \in R \setminus \{x\}}q_z(x) \ge \frac c T \right] \le \frac T c \cdot \frac \gamma T = \frac \gamma c.
    \end{equation}
    Let $J$ denote the subset of $x \in I$ that satisfy \eqref{inverter-condition-1} but not \eqref{inverter-condition-2}.
    Note that \eqref{inverter-condition-1} and \eqref{inverter-condition-2} are independent for each $x \in I$, since \eqref{inverter-condition-1} is whether $x \in R$ and \eqref{inverter-condition-2} only concerns the intersection of $R$ and $[N] \setminus \{x\}$.
    Therefore by \eqref{eq-prob-J}, the probability that $x \in I$ satisfies $x \in J$ is at most $\gamma^2/(cT^2)$.
    Hence, by Markov's inequality,
    \begin{equation}
        \label{eq-size-J}
        \Prob{R}\left[|J| \le \frac{10|I|\gamma^2}{cT^2}\right] \ge 0.9.
    \end{equation}
    From \eqref{eq-size-H} and \eqref{eq-size-J}, we get that with probability at least 0.8 over the choice of $R$,
    \begin{equation*}
        |G| = |H| - |J|
            \ge \frac{|I|\gamma}{2T^2} - \frac{10|I|\gamma^2}{cT^2}
            \ge \frac{\eps'\gamma N}{2T^2} \left(1 - \frac{5 \gamma^2} c \right)
            = \Omega\left(\frac{\eps N}{T^2}\right),
    \end{equation*}
    given that $\gamma$ is a small enough positive constant.
\end{proof}

We now proceed to describe the QRAC-VL scheme for encoding $\pi^{-1}$.
If $\pi \not\in X$ or $|G|$ is smaller than $O(\eps N/T^2)$, the encoding simply sets a (classical) flag (which takes one bit) and stores the entire permutation table of $\pi^{-1}$ (we will denote this as case A).
In this case, it is straightforward to construct a decoder that succeeds with probability 1.

Otherwise assuming $G$ is large enough, we clear the first flag, and proceed with our QRAC-VL that computes (if necessary) and outputs the following information $\beta$ as our encoding: (which we will denote as case B)
\begin{itemize}
    \item The size of $G$, encoded using $\log N$ bits;
    \item The set $G \subseteq R$, encoded using $\log \binom{|R|}{|G|}$ bits;
    \item The permutation $\pi$ restricted to input outside of $G$, encoded using $\log (N!/|G|!)$ bits;
    \item Quantum advice used by the algorithm repeated $\rho$ times $\alpha^{\bigotimes \rho}$, for some $\rho$ that we will decide later.
        (We can compute this as the encoder can preprocess multiple copies of the same advice. Note that this is the only part of our encoding that is not classical.)
\end{itemize}

Upon given the encoding $\beta$, some image $y \in [N]$, and the algorithm's randomness $R$, the decoder first proceeds to recover set $G$ and $\pi(x)$ for every $x \not\in G$.
If the given $y = \pi(x)$ for some $x \not\in G$, the decoder outputs $x = \pi^{-1}(y)$.
Otherwise, the decoder constructs $\pi'$ to be
\begin{equation*}
    \pi'(x) = \begin{cases}
        y, & x \in G; \\
        \pi(x), & x \not\in G.
    \end{cases}
\end{equation*}
Then the decoder extracts $\alpha_1, \alpha_2, ..., \alpha_\rho$, and invokes $\mathcal A^{\pi'}(\alpha_i, y)$ for each $i \in [\rho]$ and outputs their majority vote.
Let $\ket{\phi_\pi}$ and $\ket{\phi_{\pi'}}$ denote the final states of $\mathcal A$ when it is given the oracle $\pi$ and $\pi'$ respectively.
Then by \Cref{swapping-lemma} and the definition of $G$,
\begin{equation*}
    \| \ket{\phi_\pi} - \ket{\phi_{\pi'}} \|
        \le \sqrt{T\sum_{z \in R \setminus \{x\}}{q_z(x)}}
        \le \sqrt{T \cdot \frac c T}
        = \sqrt c.
\end{equation*}
As $x \in I$, by the definition of $I$, measuring $\ket{\phi_\pi}$ gives $x$ with probability at least $2/3$.
Given $c$ is a small enough positive constant, measuring $\ket{\phi_{\pi'}}$ will also give $x$ with probability at least $0.6$.

We now examine the length of our encoding.
With probability $1 - \eps'$, we have $\pi \not\in X$;
    with probability $\eps' \cdot (1 - 0.8)$, we have $\pi \in X$ but $G$ is small.
Therefore, over all, with probability $1 - 0.6 \eps$, our encoding will take case A, where the encoding consists of $1 + \log N!$ classical bits and decoder succeeds with probability 1.

With probability $0.4 \eps$, our encoding takes case B, and the size of the encoding will be
\begin{align*}
    1 + \log N + \log \binom{|R|}{|G|} + \log (N!/|G|!) + \rho S.
\end{align*}
By \eqref{eq-nondegenerate-T}, $\log \binom{|R|}{|G|} = O(|G| \log(|R|/|G|)) = O(|G| \log 1/\eps) = o(|G| \log |G|)$, and we can rewrite the size of the encoding as
\begin{align*}
    \rho S - \log |G|! + \log N! + o(\log |G|!).
\end{align*}
In this case, when the decoder is queried a point inside what she has remembered, that is $y \not\in \pi(G)$ (which occurs with probability $1 - |G|/N$), she recovers the correct pre-image with probability 1;
    otherwise, with one copy of the advice, she recovers the correct pre-image with probability 0.6, therefore with $\rho$ copies, by Chernoff's bound, she recovers the correct pre-image using majority vote, with probability $1 - \exp(-\Omega(\rho))$.

Overall, the average encoding length is at most $1/2 \cdot (\eps \rho S + |G| H(\eps) - \eps \log |G|! + \eps \log N) + \log N!$, and the average success probability is $1 - |G|/N \cdot \exp(-\Omega(\rho))$.
By setting $\rho = \Omega(\log(N/\eps)) = \Omega(\log N)$, the average success probability\footnote{Technically, we proved that the average success probability will be at least this much. However, as the success probability is monotone in encoding length, it is not hard to see that we can still use \Cref{theorem-qrac-vl-permutation}.} will be $1 - O(\eps/N)$.
By \eqref{eq-nondegenerate-eps} and \Cref{theorem-qrac-vl-permutation}, we have
\begin{equation*}
    \log N! + 1/2 \cdot (\eps \log |G|! - \eps \rho S - o(\eps \log |G|!) - \eps \log N) \ge \log N! - O(\log N).
\end{equation*}
Given \eqref{eq-nondegenerate-T}, \eqref{eq-nondegenerate-S}, i.e. $S, T$ satisfy some non-trivial conditions, we can simplify the expression above and obtain
\begin{equation*}
    \log |G|! + o(\log |G|!) \ge \Omega(S \log N).
\end{equation*}
As we are conditioning on the event that $G$ is large, plugging in the lower bound on $|G|$, we obtain that $ST^2 \ge \tilde\Omega(\eps N)$.

\section{Proof of \texorpdfstring{\Cref{functions}}{Theorem for Functions}}
Given a function inverter $(\alpha, \mathcal A)$ that inverts an $\eps$ fraction of the input.
For function $f: [M] \to [N]$, define $f^{-1}(y) = x$ if such $x$ exists, else $\bot$.
Using this notion, we can equivalently view sampling a function $f$ from $F_M$ as sampling an inverse function $f^{-1}$ from all the possible partitions of $[M]$ into $N$ bags, denoted as $P_M$.
Let $X$ sampled from $P_M$ as in \Cref{theorem-qrac-vl}, then $S(X) = M \log N$ and
\begin{align*}
	S(X_J) &= M \left(-\frac 1 N \log \frac 1 N - \left(1 - \frac 1 N\right) \log \left(1 - \frac 1 N\right) \right) \\
		&= \frac M N \left( N \log N - (N - 1) \log (N - 1) \right) \\
		&\le \frac M N (\log N + \log e).
\end{align*}

\begin{corollary}
\label{theorem-qrac-vl-function}
For any QRAC-VL for partitions $P_M$ with $\delta = 1 - \beta$ for any $\beta$, we have
\begin{equation*}
L \ge M \log N - M\beta\left(\log(e/\beta) + \frac{M}{N} (\log N + \log e)\right).
\end{equation*}
\end{corollary}

Now we construct the encoding scheme given the inverter.
Similarly as before, there is a subset $X_1 \subseteq F_M$ of size at least $0.5 \eps \cdot N^M$ such that for each function in $X_1$ the inverter is able to invert at least $\eps/2$ fraction of the input.
Let $X_2$ be functions where there exists an image in the function that has more than $K := \left(\frac{2M}N + 1\right) \cdot C \cdot \log(M/\eps) = \tilde O(1)$ pre-images for some constant $C$.
We claim that $|X_2| \le 0.1 \eps N^M$ (for cases when $M \le N$ and $M > N$, by using multiplicative form of Chernoff bound and union bound on the number of pre-images for each image.
Let $X_3 = X_1 - X_2$ with size at least $0.4 \eps N^M$, that is the set of functions that both have a large amount of invertible points and each image does not have a lot of pre-images.

Consider a function $f \in X_3$, and let $I$ be the set of indices $x \in [M]$ such that $\mathcal A$ when given input $f(x)$ returns exactly $x$ (conditioned on $f$ evaluating on the input is indeed $f(x)$) with the highest probability (ties are broken arbitrarily).
It is not hard to prove that $|I| \ge \frac{\eps M}{2K}$.
We sample a subset $R \subseteq [M]$, with each element independently chosen with probability $\gamma/T^2$ for some constant $\gamma$ that we will decide later.

Let $G \subseteq I$, where $x \in G$ if
\begin{enumerate}
    \Item \begin{equation}
        \label{function-inverter-condition-1}
        x \in R;
    \end{equation}
    \item The total query magnitude on $R \setminus \{x\}$ while running $A^{f}(\alpha, f(x))$ is bounded by $c/T$, that is,
    \begin{equation}
        \label{function-inverter-condition-2}
        \sum_{z\in R \setminus \{x\}}{q_z(x)}\le \frac c T.
    \end{equation}
\end{enumerate}

\begin{claim}
    With probability at least 0.75 over the choice of $R$, $|G| = \Omega\left(\frac{\eps M}{KT^2}\right)$.
\end{claim}
\begin{proof}
	The proof is almost exactly the same as in the case for permutations.

    Let $H = R \cap I$.
    Due to the definition of $R$, $|H|$ is distributed according to a binomial distribution.
    Therefore, the expected value of $|H|$ is $|I| \gamma/T^2$.
    By the multiplicative Chernoff bound and \eqref{eq-nondegenerate-T},
    \begin{equation}
        \label{eq-function-size-H}
        \Prob{R}\left[|H| \ge \frac{|I|\gamma}{2T^2}\right] \ge 0.95
    \end{equation}
    for all sufficiently large $N$.

    By definition, each query that $\mathcal A$ makes is of unit length.
    Since $\mathcal A$ makes at most $T$ queries, by \Cref{def-total-query-magnitude},
    \begin{equation*}
        \sum_{z \in [N]} q_z(x) \le T.
    \end{equation*}
    By linearity of expectation,
    \begin{equation*}
        \Epc{R}\left[\sum_{z\in R \setminus \{x\}}q_z(x) \right]
            = \sum_{z\in [N] \setminus \{x\}} \frac \gamma {T^2} q_z(x)
            \le \frac \gamma{T^2} T
            = \frac \gamma T.
    \end{equation*}
    Hence, by Markov's inequality,
    \begin{equation}
        \label{eq-function-prob-J}
        \Prob{R}\left[\sum_{z\in R \setminus \{x\}}q_z(x) \ge \frac c T \right] \le \frac T c \cdot \frac \gamma T = \frac \gamma c.
    \end{equation}
    Let $J$ denote the subset of $x \in I$ that satisfy \eqref{function-inverter-condition-1} but not \eqref{function-inverter-condition-2}.
    Similarly, here \eqref{function-inverter-condition-1} and \eqref{function-inverter-condition-2} are also independent for each $x \in I$, since \eqref{function-inverter-condition-1} is whether $f(x) \in R$ and \eqref{function-inverter-condition-2} only concerns the intersection of $R$ and $[N] \setminus \{f(x)\}$.
    Therefore by \eqref{eq-function-prob-J}, the probability that $x \in I$ satisfies $x \in J$ is at most $\gamma^2/(cT^2)$.
    Hence, by Markov's inequality,
    \begin{equation}
        \label{eq-function-size-J}
        \Prob{R}\left[|J| \le \frac{10|I|\gamma^2}{cT^2}\right] \ge 0.9.
    \end{equation}
    From \eqref{eq-function-size-H} and \eqref{eq-function-size-J}, we get that with probability at least 0.75 over the choice of $R$,
    \begin{equation*}
        |G| = |H| - |J|
            \ge \frac{|I|\gamma}{2T^2} - \frac{10|I|\gamma^2}{cT^2}
            \ge \frac{\eps\gamma M}{4KT^2} \left(1 - \frac{5 \gamma^2} c \right)
            = \Omega\left(\frac{\eps M}{KT^2}\right),
    \end{equation*}
    given that $\gamma$ is a small enough positive constant.
\end{proof}

We now proceed to describe the QRAC-VL scheme for encoding the partition $f^{-1}$.
If $f \not\in X_3$ or $|G|$ is not at least $\Omega(\eps M/(KT^2))$, the encoding simply sets a (classical) flag (which takes one bit) and stores the entire table of $f^{-1}$ (we will denote this as case A).
In this case, it is straightforward to construct a decoder that succeed with probability 1.

Otherwise assuming $f \in X_3$ and $G$ is large enough, we clear the first flag, and proceed with our QRAC-VL that computes (if necessary) and outputs the following information $\beta$ as our encoding: (which we will denote as case B)
\begin{itemize}
    \item The size of $G$, encoded using $\log(M + N)$ bits;
	\item The set $G \subseteq R$, encoded using $\log \binom{|R|}{|G|}$ bits;
    \item The set $f(G) \subseteq [N]$, encoded using $\log \binom N{|G|}$ bits;
    \item The function $f$ restricted to input outside of $G$, encoded using $(M - |G|) \log N$ bits;
	\item Hash tags $h_1, ..., h_{|G|}$ for each $y \in f(G)$, each of length $\log(K \log N) = \log K + \log \log N$, encoded using $|G| \cdot (\log K + \log \log N)$;
    \item Quantum advice used by the algorithm repeated $\rho$ times $\alpha^{\bigotimes \rho}$, for $\rho = \tilde O(K)$.
\end{itemize}

Upon given the encoding $\beta$, some image $y \in [N]$, and the algorithm's randomness $R$, the decoder first proceeds to recover set $G, f(G)$ and $f(x)$ for every $x \not\in G$.
If the given $y \not\in f(G)$, the decoder outputs $x = f^{-1}(y)$.
Otherwise, the decoder constructs
\begin{equation*}
    f'(x) = \begin{cases}
        y, & x \in G; \\
        f(x), & x \not\in G.
    \end{cases}
\end{equation*}
Then the decoder extracts $\alpha_1, \alpha_2, ..., \alpha_\rho$, and invokes $\mathcal A^{f'}(\alpha_i, y)$ to obtain $\rho$ outputs.
After measuring the outputs, the decoder hashes each output and compares with the hash $h_y$ in the encoding.
Finally, the decoder randomly chooses a output with the correct hash, combining other pre-images in the encoding as the output pre-image set.

Let $\ket{\phi_f}$ and $\ket{\phi_{f'}}$ denote the final states of $\mathcal A$ when it is given the oracle $f$ and $f'$ respectively.
Then by \Cref{swapping-lemma} and the definition of a good element,
\begin{equation*}
    \| \ket{\phi_f} - \ket{\phi_{f'}} \|
        \le \sqrt{T\sum_{z \in R \setminus \{x\}}{q_z(x)}}
        \le \sqrt{T \cdot \frac c T}
        = \sqrt c.
\end{equation*}
As $x \in I$, by the definition of $I$, measuring $\ket{\phi_f}$ gives some pre-image of $y$ that is in $G$ with probability at least $2/3 \cdot 1/K$.
Given $c$ is a small enough positive constant, measuring $\ket{\phi_{f'}}$ will also give $x$ with probability at least $0.6/K$.
Assuming the logarithmics in $\rho = \tilde O(K)$ is large enough, we can find at least one correct output in this process with probability at least $1 - 1/\log N$.
Due to the length of the hash tag and \Cref{thm-hash}, all the incorrect outputs will be discarded with probability $1 - 1/\log N$.
Overall, the success probability of our decoding procedure for a $y \in f(G)$ is at least $1 - 2/\log N$.

We now examine the length of our encoding.
With probability $1 - 0.6 \eps$, we have $f \not\in X_3$;
    with probability $\eps \cdot 0.4 \cdot (1 - 0.75)$, we have $f \in X$ but $G$ is small.
Therefore, over all, with probability $1 - 0.7 \eps$, our encoding will take case A, where the encoding consists of $1 + \log N!$ classical bits and decoder succeeds with probability 1.

With probability $0.3 \eps$, our encoding takes case B, and the size of the encoding will be
\begin{align*}
    1 &+ \log(M + N) + \log \binom{|R|}{|G|} + \log \binom N{|G|} + (M - |G|) \log N \\
      &+ |G| \log(K \log N) + \rho S,
\end{align*}
which is at most
\begin{align*}
    M \log N + |G| \log \frac{O(K^2 \log N)}\eps - |G| \log |G| + \rho S,
\end{align*}
for all sufficiently large $N$.
In this case, when the decoder is queried a point inside what she has remembered, that is $y \not\in \pi(G)$ (which occurs with probability $1 - |G|/N$), she recovers the correct pre-image with probability 1; otherwise, she recovers the correct pre-image with probability at least $1 - 2/\log N$.

Overall, the average success probability is at least $1 - 0.15\eps|G|/(N \log N) \le 1 - \Omega(1/N^{10})$.
By \Cref{theorem-qrac-vl-function} and $M/N + 1 = \Theta(1)$ by \eqref{eq-not-too-many-preimages}, we have
\begin{align*}
    &\mathrel{\phantom{=}} 0.3 \eps \cdot \left(|G| \log \frac{O(K^2 \log N)}\eps - |G| \log |G| + \rho S\right) \\
      &\ge - 0.15 \frac{\eps |G|M}{N \log N} \cdot O(\log N).
\end{align*}
Using the fact that \eqref{eq-function-nondegenerate-T}, \eqref{eq-function-nondegenerate-S}, we can ignore the lower order terms and obtain
\begin{equation*}
    \tilde O(|G|) \ge \tilde \Omega(SK).
\end{equation*}
Thus, $ST^2 \ge \tilde\Omega(\eps M)$.

\section{Open Questions}

Our work still does not answer whether there exists a tighter asymptotic lower bound like $ST + T^2 \ge \eps N$, nor whether there exists an attack using quantum advice that achieves $ST^2 = \eps N$.

On the other hand, it seems hard to generalize our techniques to handle random functions where $M \gg N$.
Say $M = N^2$.
It turns out that for whatever choice of $G \subseteq R$, remembering where $G$ is, and $f$ for points outside of $G$ is already too much (requires number of bits greater than $M \log N$).
Recall that $|R| \propto M/T^2$, but if we only remember one pre-image per image, $|G| \le N$.
Therefore under these parameters, $\log \binom{|R|}{|G|} \ge |G| \log N > |G| \log |G|$ and we will lose the non-trivial savings we get from the reduction.
Therefore, a natural direction would be to prove any meaningful lower bound for random function inversion under the regime where $M \gg N$.

\iflipics\else
\section*{Acknowledgements}

The authors would like to thank Nai-Hui Chia, Luca Trevisan, Xiaodi Wu, and Penghui Yao for their helpful insights during the discussions. We also thank the anonymous reviewers at QIP and ITC for pointing out various issues in the paper.
\fi

\bibliography{refs}
\end{document}